\documentclass{elsarticle}
\usepackage{times}
\usepackage{helvet}
\usepackage{courier}
\usepackage{amsfonts}
\usepackage{amsthm}
\usepackage{amsmath}
\DeclareMathOperator{\sgn}{sgn}
\DeclareMathOperator{\ds}{ds}
\DeclareMathOperator{\cs}{cs}

\begin{document}
\theoremstyle{plain}
\newtheorem{claim}{Claim}
\newtheorem{theorem}{Theorem}
\newtheorem{lemma}{Lemma}
\newtheorem{corollary}{Corollary}
\newtheorem{proposition}{Proposition}

\begin{frontmatter}

\title{Submodular Goal Value of Boolean Functions}
\author[eb]{Eric Bach}
\address[eb]{Dept. of Comp. Sciences, University of Wisconsin,
1210 W. Dayton St., Madison, WI 53706, USA} 
\author[jd]{ J\'er\'emie Dusart}
\address[jd]{The Caesarea Rothschild Institute, University of Haifa,
Mount Carmel,
Haifa 31905, Israel}
\author[lh]{Lisa Hellerstein\corref{cor1}}
\cortext[cor1]{Corresponding author}
\ead{lisa.hellerstein@nyu.edu}
\address[lh]{Dept. of Comp. Science and Engineering,
NYU School of Engineering,
6 Metrotech Center,
Brooklyn, NY 11201, USA}
\author[dk]{Devorah Kletenik}
\address[dk]{Dept. of Comp. and Information Science,
Brooklyn College,
City University of New York,
2900 Bedford Avenue,
Brooklyn, NY 11210, USA
}

\begin{abstract}
Recently, Deshpande et al.\ introduced a new measure of the complexity of a 
Boolean function.
We call this measure the ``goal value'' of the function.
The goal value of $f$ is defined in terms of a monotone,
submodular utility function associated with $f$.
As shown by Deshpande et al., 
proving that a Boolean function $f$ has small goal value can lead to a good approximation 
algorithm for the Stochastic Boolean Function Evaluation problem for $f$.
Also, if $f$ has small goal value, it indicates a close relationship 
between 
two other measures of the complexity of $f$, its average-case
decision tree complexity 
and its average-case certificate complexity.  
In this paper, we explore the goal value measure in detail.
We present bounds on the goal values of arbitrary and specific Boolean 
functions,
and present results on properties of the measure.
We compare the goal value measure to other, previously studied, 
measures of the complexity of Boolean functions.  
Finally, we discuss a number of open questions provoked by our work.
\end{abstract}

\begin{keyword}
Boolean functions \sep submodularity \sep read-once formulas
\end{keyword}

\end{frontmatter}

\section{Introduction}
Deshpande et al.\ introduced a
new measure of the complexity of a Boolean function~\cite{deshpandeHellersteinKletenik}.
We call this measure the ``goal value'' of the function.\footnote{Deshpande et al. called it the \emph{$Q$-value}, but
we have chosen to use a more intuitive name.}
In this paper, we explore properties of the goal value measure.

We begin with some background.  The goal value measure was introduced
as part of an effort to develop approximation algorithms for the
\emph{Stochastic Boolean Function Evaluation} (SBFE) problem.  In this problem,
we are given a Boolean function $f(x_1, \ldots, x_n)$, 
and we need to evaluate it on
an initially unknown
input $x$. 
For example, consider a situation where
each bit $x_i$ of $x$ corresponds to the result of a medical test,
and $f$ is a function of the test results, indicating whether the
patient has a certain disease.

In addition to $f$, the input to the SBFE problem includes
a cost $c_i$ and probability $p_i$ for each variable $x_i$.
The cost $c_i$ is the cost associated with
determining the value of $x_i$, and $p_i$ is the
independent probability that $x_i$ is 1. 
We may choose which bits to ``buy'' sequentially, and the choice of the
next bit can depend on the choices of the previous bits.
The task in the SBFE problem is to choose the sequence of bits so as to minimize the expected cost incurred
before we determine $f(x)$.

Once we know all bits of $x$, we have enough information to determine $f(x)$.
In fact, once we know enough bits that they form a certificate of $f$
(i.e., the value of $f(x)$ is already determined by those bits alone),
we also have enough information to determine $f(x)$.
Thus any set of bits that forms a certificate has the same value to us.
However, if the bits we know do not form a certificate, then we cannot determine $f(x)$. 
Nevertheless, we may still want to think of these bits as having some value towards our goal of determining $f(x)$.
We can therefore define a (non-traditional)
utility function that quantifies the value of knowing particular subsets of the bits, by assigning
a non-negative integer value to each partial assignment.
For example, the utility function might assign a value of 4 to knowing that $x_1 = 1$ and $x_3 = 0$.
The maximum value of this utility function should be achieved precisely on those sets of bits that form a certificate.
Let us call this value $Q$.

Once we have such a utility function,
there is a natural greedy approach to choosing which bits to ``buy'': 
buy bits one by one, each time buying
the bit that achieves the largest expected increase in utility value, per unit cost, until you can determine $f(x)$.
In general, this greedy approach will not perform well.
However, if the utility function has two useful properties, namely monotonity and submodularity,
then a prior result of Golovin and Krause implies that
the expected cost of the greedy approach is at most a factor $O(\log Q)$ times the optimal expected cost~\cite{golovinKrause}.

We will define monotonicity and submodularity formally below. Intuitively,
monotonicity means that utility increases (or stays the same) as more bits are acquired.
Submodularity is a diminishing returns property:
if we consider the value of a bit to be the increase in utility that it would provide,
then postponing the purchase of a particular bit cannot increase its value.

Call a utility function with the above properties a \emph{goal function} for $f$.
The \emph{goal value} of a Boolean function $f$ is the smallest possible value of $Q$, of any
goal function for $f$.


Intuitively, we can think of a goal function for $f$ as being a
monotone, submodular surrogate for $f$. It has an expanded domain
(including partial assignments to the variables) and an expanded range (the set $\{0, 1, \ldots, Q\}$, rather than just $\{0,1\}$).
The goal value of $f$ measures how much we need to expand the range in order to 
achieve both monotonicity and submodularity.
Deshpande et al.\ showed that,
in addition to having implications for Boolean function evaluation, 
proving that
a Boolean function $f$ has small goal value indicates a close relationship 
between 
two other measures of the complexity of $f$: its average-case
decision tree complexity 
and its average-case certificate complexity.


%


In this paper, however, our focus is not on the implications of goal value. 
Instead, we focus on
understanding the fundamental properties of the goal value measure itself.
We prove new bounds on the goal values of general and specific Boolean functions. 
We also consider the related 1-goal value and 0-goal value measures.
Our results suggest a number of interesting open questions.

We begin by presenting necessary definitions and explaining the known
connections between goal value, decision tree complexity, and Boolean function evaluation. 

We then prove some generic properties of goal functions.
While a single Boolean function can have different goal value
functions, we show that two distinct Boolean functions can only have
the same goal function if they are complements of each other.
An interesting open question is whether each Boolean function has
a unique {\em optimal} goal function.

We show
that the goal value of every Boolean function $f$ is
at least $n$, if $f$ depends on all $n$ of its input variables.
We note that this lower bound is tight for certain
Boolean functions, such as AND, OR and XOR functions. We also present
goal value bounds for Boolean $k$-of-$n$ functions.

Deshpande et al.\ showed that the goal value of every Boolean function
$f$ is upper bounded by $\ds(f) \cdot \cs(f)$, where $\ds(f)$ is the minimum number of terms in a DNF formula for the function, and $\cs(f)$ is the minimum number of clauses in a CNF formula for the function~\cite{deshpandeHellersteinKletenik}. 
We show that if $f$ is a Boolean read-once function, then the goal value of $f$ is equal to
$\ds(f) \cdot \cs(f)$.  Thus for read-once functions, there is a close relationship between
goal value and the traditional complexity measures $\ds(f)$ and $\cs(f)$.

We show that the goal value of any Boolean function $f$ is at most $2^n - 1$.
We do not know how close goal value can come
to this upper bound.  Deshpande et al.\ showed that two specific Boolean functions
have goal values lower bounded by
$2^{n/2}$.  
Using our bound on the goal value of read-once functions, we 
show that there is a read-once function whose goal value is $\frac{n}{3}3^{\frac{n}{3}}$,
which is equal to
$2^{\alpha n}(n/3)$, where $\alpha = \frac{\log_2 3}{3} \approx .528$. 

We give an (exponential-sized) integer program for arbitrary Boolean function $f$
whose solution is an optimal goal value function for $f$.


\section{Terminology and Notation}

\textbf{Boolean function definitions:}
Let $V = \{x_1, \ldots, x_n\}$.
A partial assignment is a vector $b \in \{0,1,*\}^n$.
For partial assignments $a,b \in \{0,1,*\}^n$, 
$a$ is an {\em extension} of $b$, written $a \succeq b$,
if $a_i = b_i$ for all $b_i \neq *$.
We also say that $b$ is {\em contained} in $a$.

Given Boolean function $f:\{0,1\}^n \rightarrow \{0,1\}$,
a partial assignment $b \in \{0,1,*\}^n$
is a {\em 0-certificate} of $f$ if
 $f(a) = 0$ for all $a \in \{0,1\}^n$ such that $a \succeq b$.
It is a {\em 1-certificate} if $f(a) = 1$
 for all $a \in \{0,1\}^n$ such that $a \succeq b$.
It is a {\em certificate} if it is either a
0-certificate or a 1-certificate.
We say that $b$ {\em contains} a variable $x_i$ if $b_i \neq *$.
We will occasionally treat a certificate $b$ as being the set of variables $x_i$ such that $b_i \neq *$,
together with the assignments given to them by $b$.
The {\em size} of a certificate $b$ is the number of variables $x_i$ that it contains.
If $b$ and $b'$ are certificates for $x$ with $b \succeq b'$, we say that $b'$ {\em is contained in} $b$.
A certificate $b$ for $f$ is {\em minimal} if there is no certificate $b'$ of $f$, such that $b' \neq b$ and
$b \succeq b'$.

The {\em certificate size} of a Boolean function $f$ is the maximum, over all $x \in \{0,1\}^n$, 
of the size of the smallest certificate contained in $x$.
Given a distribution $D$ on $\{0,1\}^n$,
the {\em expected certificate size} of $f$ is the expected value, for $x$ drawn 
from $D$, of the size of the smallest certificate contained in $x$.
 
It is well-known that every Boolean function $f$ can be 
represented as a {\em polynomial over GF(2)}.  
That is, 
either $f$ is the constant function $f = 0$, or
$f$ is computed by an expression of the form
$T_1 + T_2 + \ldots + T_k$ for some $k > 0$, where each $T_i$ is the conjunction (i.e., AND) of the
variables in some subset of $V$, no two terms contain exactly the same subset $S_i$ of variables, and addition is modulo 2.
(If $S_i = \emptyset$, then $T_i = 1$.)  This is sometimes called the ring-sum expansion
of the function.
Each Boolean function has a unique such representation,
up to permutation of the terms $T_i$.
A variable $x_i$ appears in the GF(2) polynomial for Boolean function $f$
iff function $f$ depends on variable $x_i$.

For function $f(x_1, \ldots, x_n)$ defined on $\{0,1\}^n$,
and $b \in \{0,1,*\}$,
the function $f_b$ {\em induced on $f$ by $b$} is defined as follows.
Let $V_b = \{x_i|b_i = *\}$. 
Then $V_b$ is the set of input variables for function $f_b$,
and the domain of $f_b$ is the set of assignments
that assign 0 or 1 to each variable.
For each $a$ in the domain,
$f_b(a) = f(a \backslash b)$, where
$a \backslash b$ denotes the assignment to $V$ produced by setting
the variables in $V_b$ according to $a$, and the variables in $V-V_b$
according to $b$.
For $k \in \{0,1\}$, we use $f_{x_i \leftarrow k}$ to denote the
function induced on $f$ by the partial assignment setting $x_i$ to $k$,
and all other variables to $*$.
 
The \emph{$k$-of-$n$ function} is the Boolean function on $n$ variables 
whose output is $1$ if and only if at least $k$ of its input variables are $1$.
 
A \emph{read-once formula} is a Boolean formula
over the basis of  AND and OR gates,
whose inputs are distinct variables.
The variables may be negated.
If no variables are negated, the formula is said to be monotone.

We can view a read-once formula as a rooted tree with the following
properties:
\begin{enumerate}
\item  The leaves are labeled with distinct variables from the set $\{x_1, \ldots, x_n\}$,
for some $n \geq 1$, and each such variable may appear with or without negation
\item Each internal node is labeled AND or OR
\item Each node labeled AND or OR has at least 2 children.
\end{enumerate}
A read-once function is a Boolean function that can
be computed by a read-once formula.


A \emph{maxterm} of a monotone Boolean function $f(x_1, \ldots, x_n)$ is 
the set of variables contained in a minimal 0-certificate of $f$.
(Equivalently, it is a set of variables $S$ such that setting the
variables in $S$ to 0 forces $f$ to 0, but this is not true of
any proper subset of $S$.)
A \emph{minterm} of $f$ is 
the set of variables contained in a minimal 1-certificate of $f$.
If $F$ is a Boolean formula computing a monotone
function, we say that a subset $S$ of its variables is a maxterm (minterm) of $F$
iff $S$ is a maxterm (minterm) of the function is computes.

We use $\ds(f)$ to denote the minimum number of terms in a DNF formula for $f$, and $\cs(f)$ to denote
the minimum number of clauses in a CNF formula for $f$.

A decision tree computing a Boolean function $f(x_1, \ldots, x_n)$ is a binary tree where each binary node is 
labeled by a variable in $V$, and each each leaf is labeled either 0 or 1.
The left child of an internal node labeled $x_i$ corresponds to $x_i = 0$, and the
right child corresponds to $x_i = 1$.  The tree computes $f$ if for each $x \in \{0,1\}^n$,
the root-leaf path induced by $x$ ends in a leaf whose label is $f(x)$.
Given a decision tree $T$ computing $f$, and $x \in \{0,1\}^n$, let $Depth(T,x)$
denote the number of internal nodes on the root-leaf path in $T$ induced by $x$.
Then the depth of $T$ is the maximum value of $Depth(T,x)$ for all $x \in \{0,1\}^n$.
Given a probability distribution on $\{0,1\}^n$, the {\em expected depth of $T$} is the
expected value of $Depth(T,x)$, when $x$ is drawn from the distribution.
The {\em average depth} of $T$ is its expected depth with respect to the uniform
distribution on $\{0,1\}^n$.
The {\em decision tree depth} of Boolean function $f$ is the minimum depth of any decision tree
computing $f$.  The {\em expected decision tree depth} of Boolean function $f$,
with respect to a distribution on $\{0,1\}^n$, is the minimum expected depth
of any decision tree computing $f$, with respect to that distribution.

The {\em rank} of a (binary) decision tree $T$ is defined as 
follows~\cite{ehrenfreucht}: if $T$ is a leaf, then $\mathtt{rank}(T) = 0$. Otherwise, let $\mathtt{rank}_0$ denote the rank of the 
subtree rooted at the left child of the root of $T$ and $\mathtt{rank}_1$ denote the rank of the 
subtree rooted at the right child of the root of $T$.
Then

\[\mathtt{rank}(T) = 
\begin{cases}
\max\{\mathtt{rank}_0, \mathtt{rank}_1 \}& \mbox{ if } \mathtt{rank}_0 \neq \mathtt{rank}_1 \\
\mathtt{rank}_0 + 1 & \mbox{ otherwise }
\end{cases}
\]

A \emph{$k$-decision list} is a list of pairs $(t_i, v_i)$ where $t_i$ is a term (conjunction) of at most $k$ literals and $v_i \in \{0,1\}.$ The last term $t_i$ is the empty term, which is equal to 1. A decision list $L$ defines a Boolean function such that for any assignment $x \in \{0,1\}^n$, $L(x) = v_j$ where $j$ is the minimum index such that $t_j(x) = 1$~\cite{rivestDecisionList}.

We say that a Boolean function $f$ is computed by a \emph{degree $d$ polynomial threshold function} if there is a 
multivariate polynomial $p(x_1, \ldots, x_n)$ of (total) degree $d$, 
over the real numbers, such that for all $x \in \{0,1\}^n$, $f(x) = \sgn(p(x))$, where $\sgn(z) = 1$ if $z \geq 0$ and $\sgn(z) = 0$ if $z < 0$.   
The {\em polynomial threshold degree} of $f$ is the minimum $d$ such that $f$ is computed by a polynomial threshold function of degree $d$.

\smallskip
\textbf{Submodularity Definitions:}
 Let $N = \{1, \ldots, n\}$.
In what follows, we assume that utility functions are integer-valued.
In the context of standard work on submodularity,
a {\em utility function} is a set function
$g:2^N \rightarrow \mathbb{Z}_{\geq 0}$.
Given $S \subseteq N$ and $j \in N$, $g_S(j)$ denotes
the quantity $g(S \cup \{j\}) - g(S)$. 

We also use the term {\em utility function}
to refer to a function $g:\{0,1,*\}^n \rightarrow \mathbb{ Z}_{\geq 0}$ 
defined on partial assignments.
For $l \in \{0,1,*\}$,
the quantity $b_{x_i \leftarrow l}$ denotes the partial assignment
that is identical to $b$ except that $b_i = l$. 

A utility  function $g:\{0,1,*\}^n \rightarrow \mathbb{Z}_{\geq 0}$ is 
{\em monotone} if for $b \in \{0,1,*\}^n$,
$i \in N$ such that $b_i = *$,
and $l \in \{0,1\}$, $g(b_{x_i \leftarrow l}) - g(b) \geq 0$;
in other words, additional information can not decrease utility.
Utility function $g$ is {\em submodular}
if for all $b,b' \in \{0,1,*\}^n$
and $l \in \{0,1\}$, 
$g(b_{x_i \leftarrow l}) - g(b)  \geq g(b'_{x_i \leftarrow l}) - g(b') $
whenever $b' \succeq b$ and $b_i = b'_i = *$.
This is a diminishing returns property.

We say that $Q \geq 0$ is the {\em goal value} of $g:\{0,1,*\}^n \rightarrow \mathbb{Z}_{\geq 0}$
if $g(b) = Q$ for all $b \in \{0,1\}^n$.

A {\em goal function} for  
a Boolean function $f:\{0,1\}^n \rightarrow \{0,1\}$ is a
monotone, submodular function $g:\{0,1,*\}^n \rightarrow \mathbb{Z}_{\geq 0}$, 
with goal value $Q$, such that 
for all $b \in \{0,1,*\}^n$, $g(x) = Q$ iff $b$ contains a certificate of $f$.
We define the {\em goal value of $f$} to be the minimum goal value of any goal function $g$ for $f$.

A 1-goal function for $f$ is a 
monotone, submodular function $g:\{0,1,*\}^n \rightarrow \mathbb{Z}_{\geq 0}$, 
such that 
for some constant $Q \geq 0$, $g(b)=Q$ if $b$ is a 1-certificate of $f$,
and $g(b) < Q$ otherwise.
We call $Q$ the 1-goal value of $g$.  

We define the {\em 1-goal value of $f$} to be the minimum goal value of any 1-goal function for $f$.
The definition of a {\em 0-goal function} for $f$  and the {\em 0-goal value} of $f$ are analogous.

We denote the goal value, 1-goal value, and 0-goal value of $f$ by $\Gamma(f)$, $\Gamma^1(f)$, and $\Gamma^0(f)$ respectively.

We say that a goal function $g$ for Boolean function $f$ is {\em optimal}
if the goal value of $g$ is $\Gamma(f)$.

We use $g(x_i = \ell)$ to denote the value of $g(b)$ at the 
$b \in \{0,1,*\}^n$ that has
$b_i = \ell$ and $b_j = *$ for $j \neq i$.

We define a directed graph that we call the {\em extension graph}.
It is similar to the Boolean lattice, but for partial assignments.
The vertices of this graph correspond to the elements of $\{0,1,*\}^n$.
For $a,b \in \{0,1,*\}^n$, there is an edge from $a$ to $b$ if $b$ can be produced
from $a$ by changing one $*$ in $a$ to 0 or 1.
(Equivalently, the extension graph is
the directed graph associated with a 
Hasse diagram for the extension relation $\succeq$.)
We imagine this graph is drawn so that higher vertices are more informative.
That is, we put the $2^n$ full assignments at the top, and the 
empty assignment $** \cdots *$ at the bottom.
The {\em weight} of a
partial assignment is the number of 0's and 1's that it has.
Then, the vertices of identical weight form a level in this graph,
and there are $n+1$ levels.

\section{Previous Results on Goal Value}
\label{fun}

Before showing our new results, we review previous results on goal value.

\subsection{Goal value and DNF/CNF size}

The following fundamental facts were shown by Deshpande 
et al.~\cite{deshpandeHellersteinKletenik}.

\begin{proposition}
\label{bipartiteOR}
For any Boolean function $f$, 
$\Gamma(f) \leq \Gamma^1(f) \cdot \Gamma^0(f)$, $\Gamma^0(f) \leq \ds(f)$, and 
$\Gamma^1(f) \leq \cs(f)$.
\end{proposition}

To see that $\Gamma^1(f) \leq \cs(f)$, let $\phi$ be a CNF formula
for $f$ with $\cs(f)$ clauses,  and let $g_1$ be the 1-goal value
function with goal value $\cs(f)$, such that $g_1(b)$ equals
the number of clauses of $\phi$ satisfied by $b$.  
The inequality $\Gamma^0(f) \leq \ds(f)$ is shown analogously 
using a 0-goal function $g_0$ such that $g_0(b)$ is equal to the
number of terms falsified by $b$ in a DNF for $f$ with $\ds(f)$ terms.

The proof that
$\Gamma(f) \leq \Gamma^1(f) \cdot \Gamma^0(f)$ relies
on a standard ``OR'' construction used previously by
Guillory and Blmes ~\cite{guilloryBilmesuai11}.
Let $g_1$ be a 1-goal function for $f$ and $g_2$ be a 0-goal function for $f$.
Let $Q_1$ and $Q_0$ be the goal values for these functions.
Then function $g(b) = Q_1 Q_2 - (Q_1 - g_1(b))(Q_2 - g_0(b))$
is a monotone, submodular function where $g(b) = Q_1 Q_2$
iff $g_1(b) = Q_1(b)$ or $g_0(b)=Q_0(b)$.

$\Gamma(f)$ can be exponential in $n$~\cite{deshpandeHellersteinKletenik}.
\subsection{Bounds on Decision Tree Depth}
Deshpande et al.\ also showed that the goal value of a Boolean function $f$ could be used to upper bound the expected
depth of a decision tree computing $f$, with respect to the uniform distribution on $\{0,1\}^n$.

Let $\mathsf{DT}(f)$ denote the expected decision tree depth of $f$,
and $\mathsf{CERT}(f)$ denote the expected certificate size of $f$,
where both expectations are with respect to the uniform distribution on $\{0,1\}^n$. 
Since every decision tree leaf is a certificate,
$\mathsf{CERT}(f) \leq \mathsf{DT}(f)$.
Let $g$ be an optimal goal function for $f$. 
Under the uniform distribution, when run on $g$,
the Adaptive Greedy algorithm of Golovin and Krause outputs a solution 
(corresponding to a decision tree) that is within a factor of $2(\ln \Gamma(f) + 1)$ of the expected certificate size of $f$. It follows that
\[\mathsf{DT}(f) \leq (2 \ln \Gamma(f) + 1) \mathsf{CERT}(f).\]

Thus if $\Gamma(f)$ is small, expected certificate size and expected depth are close to each other.
However, $\Gamma(f)$ can be very large.  In ~\cite{allenHellersteinUnluyurt}, it was
shown that the gap between $\mathsf{DT}(f)$ and $\mathsf{CERT}(f)$ can be exponential:
there is a function $f$ such that for any constant $0 < \epsilon < 1$, $\frac{\mathsf{DT}(f)}{\mathsf{CERT}(f)} = \Omega(2^{\epsilon \mathsf{CERT}(f)}$). This stands in marked contrast to results for \emph{worst-case} depth of decision tree and worst-case certificate size; it is known that the worst-case depth of a decision tree is at most quadratic in the size of the worst-case certificate~\cite{buhrmanWolf}. 


\section{Relations between Functions, Goal Functions, and Goal Values}

Clearly, any 1-goal function for $f$ determines $f$.  A similar
statement holds for 0-goal functions.  It is also evident that any 
goal function for a Boolean function determines all of its certificates.
It is interesting that a much stronger result is true.

\begin{theorem}
\label{reverse}
Any goal function determines the Boolean function, up to complementation.
That is, if $g$ is a goal function for a Boolean function $f$, 
then $g$ is a goal function for exactly one other Boolean function,
namely $\neg{f}$.
\end{theorem}

\begin{proof}
Let $g$ be a goal function for some Boolean function 
$f:\{0,1\}^n \rightarrow \{0,1\}$.
Let the goal value of $g$ be $Q$.  Since every 
assignment in $\{0,1\}^n$ is a certificate
of $f$, $g(a) = Q$ for all $a \in \{0,1\}^n$.
Now, consider two assignments in $\{0,1\}^n$ that differ by
the setting of one bit, say:
\begin{eqnarray*}
        a & = ...0... \\
        b & = ...1...
\end{eqnarray*}

The partial assignment with that bit erased is
\begin{eqnarray*}
        c & = ...*... \cr
\end{eqnarray*}
which is right below $a$ and $b$ in the
extension graph.  If $c$ has value $Q$, then $f(a) = f(b)$
and if $c$ has value $<Q$, we know that $f(a) \ne f(b)$.  This allows
us to classify all edges $a-b$ of the Boolean hypercube as one where $f$
takes the same values on the endpoints, or one where different
values are taken.  Therefore, the following ``is you is or is you ain't''
procedure can be used to recover $f$ (up to complementation):
\smallskip
{
\parindent 0.5 in
\obeylines
  There are two colors, red and white.
  Choose any corner of the hypercube, and color it red (say).
  Now do breadth first search starting with this corner.
     \qquad When a ``new'' vertex $w$ is discovered, it will be from
     \qquad\quad a previously colored vertex $v$.  Use the classification of
     \qquad\quad $(v,w)$ to color $w$.
}
\smallskip
At the end of the procedure, all vertices are colored, and we know
that there is a value $x$ (0 or 1) such that
such that $f(v) = x$ (resp. $\bar x$) for all red (resp. white)
vertices $v$. 
\end{proof}

To illustrate this theorem, we start with a goal function
for OR.  Here, we call the four full assignments $A,B,C,D$.
Edges in the graph are not shown.
$$
\begin{array}{cccccl}
        A   &   B  &    &   C     &   D   &                              \\
            &      &    &         &       &                              \\
       (2)  &  (2) &    &   (2)   &   (2) & \leftarrow \hbox{value}      \\
       00   &  10  &    &   01    &   11  & \leftarrow \hbox{assignment} \\
            &      &    &         &       &                              \\
       (1)  &  (2) &    &   (1)   &   (2) &                              \\
       *0   &  *1  &    &   0*    &   1*  &                              \\
            &      &    &         &       &                              \\
            &      & (0)&         &       &                              \\
            &      & ** &         &       &  
\end{array}
$$
Inspecting this labeled graph we see that
$$
          f(A) \ne f(B), \qquad f(C) =  f(D), \qquad f(B) =  f(D) ,
$$
so $A$ is mapped to one Boolean value
and the other assignments are mapped to the complementary value.
Thus $f$ is either OR or not-OR.  

Observe that the function recovery procedure need not examine
the entire goal function.  It can query one top level value
(to learn $Q$) and then $2^n - 1$ values at the next level down, 
corresponding to the partial assignments setting
a designated bit to $*$, and the remaining bits to values in $\{0,1\}$.

There can be more than one goal function for a Boolean function,
even if all goal functions are normalized to start at 0.  
For example, if $g$ is a goal function for $f$,
and $k$ is any positive integer,
let 
$g':\{0,1,*\}^n \rightarrow \mathbb{Z}_{\geq 0}$ be such
that $g'(b) = g(b)$ for $b = (*,*, \cdot, *)$,
and $g'(b) = g(b) + k$ otherwise.
Then $g'$ is also a goal function for $f$.

We can also multiply any goal function by a positive integer
constant, and get another one.  For this reason, the values
(extension graph labels) taken by an optimal goal function must
be relatively prime.


A more interesting question to ask,
then, is this: can a Boolean function have more than one 
{\it optimal} goal function?

A related question is to determine which Boolean functions
have equal goal values.
It is easy to prove the following proposition, which says
that goal value is invariant under some
simple transformations.  

\begin{proposition}
\label{negate}
Let $f$ be a Boolean function. If $f' = \neg f$, then
$\Gamma(f') = \Gamma(f)$.  If
$f'$ is any function produced from $f$ by replacing an input variable 
of $f$ by its negation, then $\Gamma(f') = \Gamma(f)$.
Finally, if $f'$ is produced from $f$ by permuting the input variables of $f$,
then $\Gamma(f)=\Gamma'(f)$.
\end{proposition}


Let us call
two Boolean functions (both defined on $n$ variables) C-equivalent if one can be turned into
the other by complementing inputs and/or the output.  
By Proposition~\ref{negate}, C-equivalent
Boolean functions have the same goal value.  
However, the converse is not true.  For example,
we show below (Section~\ref{lb}) that
AND and XOR have goal value $n$, but they are not C-equivalent,
because C-equivalent functions have the same Hamming weight mod 2.

We note that a formula for the number of C-equivalence classes
appears in the literature \cite[p. 152]{harrison}:
$$
\frac{ 2^{2^n} + (2^n - 1)2^{2^{n-1}+1} }
     { 2^{n+1} } .
$$
This is sequence A000133 in OEIS. It begins
$2, 5, 30, 2288, 67172362, \ldots$\ .
There is an interesting heuristic interpretation for this.
The leading term comes from taking the total number of functions, 
$2^{2^n}$, and dividing by the size of the symmetry group, $2^{n+1}$.


Now let us call two Boolean functions (both defined on $n$ variables) G-equivalent if they
have the same goal value.
It is an open question to determine the number of equivalence class of
Boolean functions with respect to G-equivalence.


\section{Lower Bound on Goal Value}
\label{lb}
In this section, we seek to lower bound the goal value for all Boolean functions. 
Let $f:\{0,1\}^n \rightarrow \{0,1\}$ be a Boolean function depending on all its variables.

We begin with the following lemma:
\begin{lemma}
\label{propnonzero}
Let $f(x_1, \ldots, x_n)$ be a Boolean function.
Let $\ell,k \in \{0,1\}$.  
Let $g$ be a goal function for $f$.
If there is a minimal certificate of $f$
setting $x_i = \ell$, then 
$g(x_i=\ell) - g(*, \ldots, *) \geq 1$. 

Similarly, if $g'$ is a $k$-goal function for $f$, and
there is a minimal $k$-certificate for $f$ setting
$x_i = \ell$, which has size $s$, then
$g'(x_i=\ell) - g'(*, \ldots, *) \geq 1$.
Further, if the size of that certificate is $s$,
then the goal value of $g'$ is at least $s$.
\end{lemma}

\begin{proof}
Let $g$ be a goal function for $f$, and let $Q$ be its goal value.
Let $g'$ be a $k$-goal function for $f$, and let $Q'$ be its $k$-goal value.
Let $d$ be a minimal $k$-certificate for $f$. Thus $g(d) = Q$ and $g'(d) = Q'$.

Let $s$ be the size of certificate $d$,
and without loss of generality, assume that $x_1, \ldots, x_s$ are the variables contained in $d$.
Consider the sequence $d^0, d^1, \ldots, d^s$ where $d^0 = (*, \ldots, *)$ and
$d^i = (d_1, d_2, \ldots, d_i, *, \ldots, *)$ for $1 \leq i \leq s$.
Consider a fixed $i$ such that $1 \leq i \leq s$.
By monotonicity, $g(d^{i-1}) \leq g(d_i)$.
Suppose $g(d^{i-1}) = g(d_i)$.  Let $\hat{d}$ be the partial assignment
that is obtained from $d^s$ by setting $x_i$ to $*$.
Since $d$ is a minimal certificate for $f$, $\hat{d}$ is not a certificate,
and so $g(\hat{d}) < Q$, and thus $g(d) - g(\hat{d})  \geq 1$.
Since $d$ and $\hat{d}$ differ only in the assignment to $x_i$,
and $\hat{d} \succeq d^{i-1}$,
by submodularity, 
$g(d^i) - g(d^{i-1})  \geq 1$ and $g(x_i = d_i) - g(*, \ldots, *) \geq 1$.

The same argument shows that
$g'(d^i) - g'(d^{i-1})  \geq 1$ and $g'(x_i = d_i) - g(*, \ldots, *) \geq 1$.
From the first inequality, we get that $g'(d) \geq s$.
The lemma follows.
\end{proof}

The proof of the above lemma also shows that $\Gamma(f) \geq d$,
where $d$ is the size of a minimal certificate for $f$.
But $d \leq n$, and we will prove that $\Gamma(f) \geq n$.
We use the following lemma, which is Theorem 1 in~\cite{salomaa}, and was also reproven in~\cite{ehrenfeucht}.
 
\begin{lemma}~\cite{salomaa}
\label{fixone}
Let $f(x_1, \ldots, x_n)$ be a Boolean function depending on all $n$ of its variables.
Then there exists a variable $x_i$ and a value $\ell \in \{0,1\}$ such that the induced function $f_{x_i \leftarrow \ell}$
depends on all $n-1$ of its input variables.
\end{lemma}

We now have the following theorem:

\begin{theorem}
\label{lowerbn}
Let $f:\{0,1\}^n \rightarrow \{0,1\}$ be a Boolean function. 
If $f$ depends on exactly $n'$ of its input variables, then
$\Gamma(f) \geq n'$.  Further, $\Gamma^1(f) + \Gamma^0(f) \geq n'+1$.
\end{theorem}

\begin{proof}
We prove the theorem in the case that $f$ depends on all $n$ of its input variables, so $n'=n$.
This is without loss of generality, because
if $n' < n$, we can consider instead the restriction of $f$ to the variables on which it depends.
It is easy to verify that this causes no change in the values of $\Gamma$, $\Gamma^1$, and $\Gamma^0$.

By repeated application of Lemma~\ref{fixone}, it follows that
there exists a sequence of partial assignments, $b^0, b^1, \ldots, b^{n'}$ in $\{0,1,*\}^n$, such that
$b^0$ is the all-$*$ assignment, $b^{n}$ has no $*$'s, each $b_{i+1}$ is an extension of $b^i$ produced by changing exactly one
$*$ to a non-$*$ value, and the function $f^i$ that is induced from $f$ by partial assignment $b^i$
depends on all $n-i$ of its variables.

Without loss of generality, assume $b^i$ assigns values in $\{0,1\}$ to variables $x_{n-i+1}, \ldots, x_n$,
and sets variables $x_1, \ldots, x_{n-i}$ to $*$.
Thus $f^i:\{0,1\}^{n-i} \rightarrow \{0,1\}$.

Let $g$ be a goal function for $f$. 
For $i \in \{1,\ldots, n\}$, let $g^i:\{0,1,*\}^{n-i} \rightarrow \mathbb{Z}_{\geq 0}$ be the
function induced on $g$ by $b^i$.
Thus for $a \in \{0,1,*\}^{n-i}$,
$g^{i}(a) = g(a \backslash b^i)$, where $a \backslash b^i$ denotes the assignment produced by setting
the variables in $\{x_1, \ldots, x_{n-i}\}$ according to $a$, and the remaining variables of $g$ according to $b^i$.
From the fact that $g$ is a goal function for $f$, it easily follows that
$g^i$ is a goal function for $f^i$, with goal value equal to the goal value of $g$.

Let $x_{n-i+1} \leftarrow \ell$, where $\ell \in \{0,1\}$, be the one-variable assignment that extends $b^{i-1}$ to produce $b^i$.
Since $f^{i-1}$ depends on $x_{n-i+1}$, 
there is an assignment $b \in \{0,1\}^{n-i+1}$ 
such that $f^{i-1}(b_{x_{n-i+1}} \leftarrow 0) \neq f^{i-1}(b_{x_{n-i+1}} \leftarrow 1)$.
Hence there
is a minimal 0-certificate or a minimal 1-certificate for $f^{i-1}$ setting $x_{n-i+1}$ to $\ell$.
Thus by Lemma~\ref{propnonzero}, $g^{i-1}(x_{n-i+1}=\ell) - g^{i-1}(*,\ldots,*) \geq 1$,
and hence $g(b^i) - g(b^{i-1}) \geq 1$.
Thus the value of $g$ increases by at least
1 for each $b^i$ in the sequence $b^0, b^1, \ldots, b^n$, and so $g(x_n) \geq n$.  This proves that 
$\Gamma(f) \geq n$.

We now need to prove that
$\Gamma^1(f) + \Gamma^0(f) \geq n+1$. 
Let $g'$ be a 1-goal function for $f$, and let 
$g''$ be a 0-goal function for $f$.
Considering the sequence $b^0, \ldots, b^n$ again, the above argument, together with
Lemma~\ref{propnonzero}, shows that for each $i \in \{1,\ldots, n-1\}$, $g'(b^i) - g'(b^{i-1}) \geq 1$ or
$g''(b^i) - g''(b^{i-1}) \geq 1$.
Thus $g'(b^{n-1}) + g''(b^{n-1}) \geq n-1$.

Either the
assignment $x_1 = b^n_1$ is in a 1-certificate for $f^{n-1}$, or in a 0-certificate (or both).
Without loss of generality, assume it is in a 1-certificate.  Then assignment $x_1 = \neg{b^n_1}$ is in a 0-certificate
for $g^{n-1}$.  Let $q'$ be the goal value of $g'$ and $q''$ be the goal value for $g''$.
Since $b^{n-1}$ contains neither a 1-certificate nor a 0-certificate for $f$,
$g'(b^{n-1}) \leq q'-1$ and 
$g''(b^{n-1}) \leq q''-1$. 
It follows that $n-1 \leq (q'-1) + (q'' - 1)$, and so
$q' + q'' \geq n+1$.  The theorem
 follows.
\end{proof}

Furthermore, these bounds can be tight, as illustrated in the following two propositions:
\begin{proposition}
\label{andorxor}
The AND, OR, and XOR functions, and their negations, have goal value $n$.
\end{proposition}
\begin{proof}
Their goal values must be at least $n$, by Theorem~\ref{lowerbn}. 
For AND, consider utility function $g$ whose value equals $n$
on any partial assignment having at least one 0, and 
whose value equals the number of 1's otherwise.
This is a goal function for AND,  and there is a dual goal function for OR. 
For XOR, consider the goal function whose value on a partial assignment is
equal to the number of variables set to 0 or 1.
\end{proof}

\begin{proposition}
\label{prop:k-of-n}
Let $f:\{0,1\}^n \rightarrow \{0,1\}$ be a Boolean $k$-of-$n$ function.
Then $\Gamma^1(f) = k$ and $\Gamma^0(f) = n - k + 1$.
\end{proposition}
\begin{proof}
For any $k$-of-$n$ function $f$, define the 1-goal function $g'$ whose value on any subset of bits
is just the minimum of $k$ and the number of bits set to 1.  This 1-goal function has goal value $k$, 
so $\Gamma^1(f) \leq k$
and
by Lemma~\ref{propnonzero}, $\Gamma^1(f) \geq k$. 
Similarly, define the 0-goal function $g''$ whose value on any subset of bits is
the minimum of $n-k+1$ and the number of bits set to 0.
\end{proof}


We can further analyze the goal value for $k$-of-$n$ functions:
\begin{theorem}
\label{thm:thresh}
Let $f:\{0,1\}^n \rightarrow \{0,1\}$ be a Boolean $k$-of-$n$ function.
Then $\Gamma(f) = k(n-k+1)$.
\end{theorem}


\begin{proof}
Let $g$ be a goal function for $f$ with goal value $Q$.
Consider a minimal 1-certificate $b$ for $f$.  It has exactly $k$ 1's and no 0's.
Since $g(b) = Q$ it follows from the submodularity and monotonicity of $g$ that
there must be at least one index
$i_1$ such that $b_{i_1} = 1$ and
$g(\{b_{i_1}\}) \geq \frac{1}{k} Q$.
Let $t_1 = g(\{b_{i_1}\})$.
Thus  $g(b) - g(\{b_{i_1}\}) = Q-t_1$.  
It follows again from the monotonicity and submodularity of $g$
that for some index $i_2 \neq i_1$ where $b_{i_2} = 1$,
$g(\{b_{i_1}, b_{i_2}\}) \geq t_1 + \frac{1}{k-1}(Q-t_1)$.
Since $t_1 \geq \frac{1}{k}Q$, and
$t_1 + \frac{1}{k-1}(Q-t_1) = t_1(1 - \frac{1}{k-1}) + \frac{1}{k-1}Q$,
we get that
$g(\{b_{i_1}, b_{i_2}\}) \geq \frac{1}{k}(1-\frac{1}{k-1})Q+\frac{1}{k-1}Q$
so
$g(\{b_{i_1}, b_{i_2}\}) \geq Q(\frac{2}{k})$.
Settng $t_2 = g(\{b_{i_1}, b_{i_2}\})$, we similarly get 
$i_3$ with
$g(\{b_{i_1}, b_{i_2}, b_{i_3}\}) \geq \frac{3}{k}Q$ and so forth,
until we have $g(\{b_{i_1}, \ldots, b_{i_{k-1}}\}) \geq \frac{k-1}{k}Q$,
where $b_{i_1}, \ldots, b_{i_{k-1}}$ are
$k-1$ of the $k$ bits  of $b$ that are set to 1.
Let $b_{i_k}$ be the remaining bit of $b$ that is set to 1.
Let $c$ be the minimal 0-certificate for $f$ such that 
$c_{i_k} = 0$, and further,
$c_i = 0$ for all $i$ where $b_i = *$, and $c_i = *$
for all other $i$.
Let $l=n-k+1$, which is the number of 0's in $c$.

Let $t_{k-1} = g(\{b_{i_1}, \ldots, b_{i_{k-1}}\})$, so
$t_{k-1} \geq \frac{k-1}{k}Q$.
Let $d$ be the assignment such that $d_{i_j} = 1$ 
for $j \in \{1,\ldots, k-1\}$, and $d_i = 0$ for all other $i$.
Then $g(d) = Q$ and $g(d) -  g(\{b_{i_1}, \ldots, b_{i_{k-1}}\}) = Q-t_{k-1}$.
Let $Q' = Q-t_{k-1}$.

Let $l=n-k+1$.
It follows from an argument similar to the one above that
there are bits $c_{j_1}, \ldots, c_{j_{l-1}}$ equal to 0
such that $g(\{b_{i_1}, \ldots, b_{i_{k-1}}, c_{j_1}, \ldots, c_{j_{l-1}}\}) \geq t_{k-1} + \frac{l-1}{l}Q'$.
Let $d'$ be the partial assignment such that
$b_{i_1}, \ldots, b_{i_{k-1}}=1$, $c_{j_1}, \ldots, c_{j_{l-1}}=0$
and $b_i = *$ for the remaining one variable $x_i$.

Thus we have $g(d') \geq t_{k-1} + \frac{l-1}{l}Q' = t_{k-1} + \frac{l-1}{l}(Q-t_{k-1})$
and so $g(d') \geq t_{k-1}(1-\frac{l-1}{l}) + \frac{l-1}{l}Q$.
Using the fact that $t_{k-1} \geq \frac{k-1}{k}Q$ we get that
$g(d') \geq \frac{k-1}{k}(1-\frac{l-1}{l})Q + \frac{l-1}{l}Q$ and so
$g(d') \geq \frac{kl-1}{kl} Q$.
Further, since $d'$ is neither a 0-certificate nor a 1-certificate for $f$,
$g(d') < Q$, and thus $g(d') \leq Q-1$.
It follows that $Q \geq kl=k(n-k+1)$. 

Finally, it follows from
Proposition~\ref{bipartiteOR} that
there exists a goal function $g$ for $f$ with goal value exactly equal to $kl = k(n-k+1)$.
\end{proof}

We can obtain goal values of some additional functions 
by applying Proposition~\ref{negate}, 
which states that goal value is unchanged under negation of variables.
Fix $b \in \{0,1\}^n$, and consider the function
$$
f(x) = \begin{cases}  1 & \mbox{if } x =  b  \\
       0 & \mbox{if } x \ne b \end{cases} 
$$
having the unique satisfying assignment $b$.
Since $f$ can be produced from the AND function by
replacement of input variables by their negations, 
the goal value of $f$ is also $n$.

Similarly, for fixed $b \in \{0,1\}^n$, consider the function
$f$ whose value is 1 on input $a \in \{0,1\}^n$ iff $a$ and $b$
agree in at least $k$ bits.  Since this function can be produced
from the $k$-of-$n$ function by replacement of variables by their negations, 
the goal value of this function is $k(n-k+1)$.

\section{Goal Value of Read-Once Functions} 

In this section we prove that a read-once function's goal value
is the product of its DNF and CNF sizes.  That is,
$\Gamma(f) = \ds(f) \cdot \cs(f)$.

Given a set of variables $T$ of a read-once formula $f$,
let $f_{T=0}$ denote the induced read-once formula produced
by setting all the variables in $T$ to 0 (and simplifying
in the standard way to eliminate newly irrelevant variables
and gates that are no longer needed).
Define $f_{T=1}$ analogously.

We will need the following lemma.

\begin{lemma}
Let $f$ be a monotone read-once function computed by a read-once
formula that contains at least two variables.
Then at least one of the following holds:

\begin{enumerate}
\item There is a maxterm $T$ of
$f$, and a partition of $T$ into two sets, $T_1$ and $T_2$, 
such that
$\cs(f_{T_1=0}) = \cs(f_{T_2=0}) = \cs(f)$.

\item
There is a minterm $S$ of
$f$, and a partition of $S$ into two sets, $S_1$ and $S_2$, 
such that
$\ds(f_{S_1=1}) = \ds(f_{S_2=1}) = \ds(f)$.
\end{enumerate}
\end{lemma}

\begin{proof}
If Condition 1 holds for $f$, we say it has a good
maxterm partition (with sets $T$, $T_1$, and $T_2$) and
if Condition 2 holds we say that
$f$ has a good minterm partition (with sets $S$, $S_1$, and $S_2$),

Let $n$ be the number of variables of $f$.
For the base case, assume $n=2$.
Then $f$ is either the AND of two variables, or the OR of two variables.
In the former case, $f$ clearly has a good minterm partition,
and in the latter, it has a good maxterm partition.

Suppose the theorem holds when $f$ is computed by a read-once formula with $n$
variables, for some $n \geq 2$.
Let $f$ be a read-once function expresed by a read-once formula with $n+1$ variables.
In an abuse of notation, we will use $f$ to denote both the function
and the formula.

We now show that $f$ either has a good maxterm or a good minterm partition.
Represent $f$ as a binary tree whose internal nodes
are all labeled AND and OR.

\smallskip
{\bf \noindent Case 1:} The root of the tree $f$ is AND.

Since the root of $f$ is AND, we have $\cs(f)=\cs(f')+\cs(f'')$
and $\ds(f)=\ds(f')\cdot\ds(f'')$.

Let $f'$ and $f''$ be the two subtrees of the root.
Since $n+1 > 2$, at least one of these is not a single variable,
so by the inductive hypothesis, it has 
a good minterm partition or a good maxterm partition.
We first consider the subcase where at least one of them has
a good maxterm partition.  Wlog, assume that $f'$ has a good
maxterm partition, with associated sets $T', T'_1$ and $T'_2$.
Then let $T = T'$, $T_1 = T'_1$, and $T_2 = T'_2$.
Since the root of $f$ is AND, $T$ is a maxterm of $f$.
Further, $\cs(f_{T_1=0}) = \cs(f'_{T_1=0}) + \cs(f'') = \cs(f')+\cs(f'') = \cs(f)$.
Similarly,
$\cs(f_{T_2=0}) = \cs(f)$. Therefore, $T$, $T_1$, and $T_2$
form a good maxterm partition for $f$.

Otherwise, at least one of the two formulas, $f'$ and $f''$,
has a good minterm partition.  Without loss of generality,
suppose $f'$ has a good minterm partition,
with associated sets $S'$, $S'_1$, and $S'_2$.
Suppose first that $f''$ is a single variable $y$.
Let $S = S' \cup \{y\}$.
$S_1 = \{y\} \cup S'_1$
and $S_2 = S'_2$.
Since $\{y\}$ is a minterm of $f''$ and $S'$ is a minterm of $f'$,
$S$ is a minterm of $f$.
Further, since $\ds(f)=\ds(f')\cdot\ds(f'')$ and $f''$ is a single variable, 
$\ds(f)=\ds(f')\cdot 1 = \ds(f')$.
Setting $y$ to 1 causes $f$ to become equal to $f'$.
Therefore $\ds(f_{S_1 = 1}) = \ds(f'_{S'_1=1}) = \ds(f')=\ds(f)$.
Also, 
$\ds(f_{S_2=1}) = \ds(f'_{S_2=1})\cdot\ds(f'') = \ds(f') \cdot 1=\ds(f)$.
Thus $S, S_1$ and $S_2$ form a good minterm partition 
for $f$. 

Again assuming that $f'$ has a good minterm partition, now
suppose that $f''$ is not a single
variable. In this subcase, by the inductive hypothesis, $f''$ has 
either a good maxterm partition or a good minterm partition.
If $f''$ has a good maxterm partition, we already showed above
that $f$ also has a good maxterm partition.
So assume that $f''$
has a good minterm partition.
Let 
$S'$, $S'_1$, and $S'_2$ be the sets for the good minterm partition
of $f'$,
and let $S'', S''_1$ and $S''_2$ be the sets for the good
minterm partition of $f''$.
Let $S = S' \cup S''$.
Let $S_1 = S'_1 \cup S''_1$ and
$S_2 = S'_2 \cup S''_2$. 
Then because the root of $f$ is AND,
we have
$\ds(f'_{S_1=1}) = \ds(f'_{S'_1=1})\cdot\ds(f''_{S''_1=1})$
and
$\ds(f'_{S_2=1}) = \ds(f'_{S'_2=1})\cdot\ds(f''_{S''_2=1})$.
Because $S', S'_1$, and $S'_2$
are a good minterm partition for $f'$ and $S''$, $S''_1$, and $S''_2$ 
are a good minterm partition for $f''$,
we have $\ds(f'_{S'_1=1})\cdot\ds(f''_{S''_1=1}) = \ds(f')\cdot\ds(f'')$
and $\ds(f'_{S'_2=1})\cdot\ds(f''_{S''_2=1}) = \ds(f')\cdot\ds(f'')$
It follows that 
$\ds(f_{S_1=1}) = \ds(f')\cdot\ds(f'')=\ds(f)$
and
$\ds(f_{S_2=1}) = \ds(f')\cdot\ds(f'')=\ds(f)$.
Thus $S, S_1$ and $S_2$ form a good minterm partition for $S$.

\smallskip
{\bf \noindent Case 2:} The root of the tree $f$ is OR.

This case follows from Case 1 by duality.
\end{proof}

We are now ready to prove the main theorem of this section.

\begin{theorem}
\label{roftheorem}
If $f$ is a read-once function, then $\Gamma(f)=\ds(f)\cdot \cs(f)$
\end{theorem}

\begin{proof}
Without loss of generality, assume $f$ is monotone.
We will also use $f$ to denote a monotone read-once formula computing $f$.

We prove the theorem
by induction on the number $n$ of variables of the function $f$.
The theorem clearly holds when $n=1$.

Assume the theorem holds when the number of variables is at most $n$.
Now suppose $f$ has $n+1$ variables.

By Proposition~\ref{bipartiteOR},
$\Gamma(f) \leq \ds(f)\cdot\cs(f)$.

We will show that
$\Gamma(f) \geq \ds(f)\cdot\cs(f)$. 

By the previous lemma, $f$ must have either a good maxterm
partition or a good minterm partition.
We give the proof only for the case where it has a good maxterm
partition.  A dual proof works in the case where it has
a good minterm partition.

So suppose $f$ has a good maxterm partition, with
$T$, $T_1$ and $T_2$.
Then

\begin{equation}
\label{cs}
\cs(f)=\cs(f_{T_1=0})=\cs(f_{T_2=0}).
\end{equation}

Since $f$ is a monotone formula, $\ds(f)$ is equal to the number
of minterms of $f$.  
The minterms of 
$f_{T_1=0}$ are precisely those minterms of $f$ that do not contain
any variables of $T_1$.
The corresponding statement holds for $T_2$.

We now exploit a well-known property of read-once formulas that is easy to prove: 
if $a$ is a minimal 0-certificate of a read-once function,
and $b$ is a minimal 1-certificate of the same function,
then $a$ and $b$ contain exactly one variable in common
(cf. ~\cite[Chap. 10]{cramaHammer}). 
Therefore,
each minterm of $f$ contains exactly one variable of the maxterm
$T$ in question. 
It follows that each minterm of $f$ 
either contains a single variable
of $T_1$ or a single variable of $T_2$.
Therefore,

\begin{equation}
\label{ds}
\ds(f) = \ds(f_{T_1=0})+\ds(f_{T_2=0}).
\end{equation}

Let $g$ be a goal function of $f$ whose goal value is $\Gamma(f)$.
Define $g_0$ to be the function of subsets $X$ of the variables of $f$
such that $g_0(X)$ is equal to the value of $g$ on the assignment
setting the variables in $X$ to 0.

Consider the function $f_{T_1=0}$.  Since it is read-once,
and has fewer than $n+1$ variables, by induction we have
$\Gamma(f_{T_1=0})=\ds(f_{T_1=0})\cdot\cs(f_{T_1=0})$.
If we take $f$ and set the variables in $T_1$ to 0, the induced
function is $f_{T_1=0}$.
It follows that we can define a goal function $h$ for $f_{T_1=0}$
whose value on any assignment $b$ to the variables of $f_{T_1=0}$
is equal to $g(b') - g_0(T_1)$, where
$b'$ is the extension of $b$ to the variables of $f$ that is
produced by setting the variables of $T_1$ to 0.

The goal value of $h$ is $\Gamma(f) - g_0(T_1)$.
By the definition of $\Gamma$, the goal value of $h$
must be at least 
$\Gamma(f_{T_1=0})$.  Since 
$\Gamma(f_{T_1=0}) =\ds(f_{T_1=0})\cdot\cs(f_{T_1=0})$,

\begin{equation}
\label{t1}
g_0(T_1) + \ds(f_{T_1=0})\cdot\cs(f_{T_1=0}) \leq \Gamma(f).
\end{equation}

The same argument with $T_2$ shows that
\begin{equation}
\label{t2}
g_0(T_2) + \ds(f_{T_2=0})\cdot\cs(f_{T_2=0}) \leq \Gamma(f).
\end{equation}

Since $T_1 \cup T_2$ is a maxterm of $f$, and $g$
is a goal function for $f$ with goal value $\Gamma(f)$, we have
$g_0(T_1 \cup T_2) = \Gamma(f)$.

By the submodularity of $g$,
\begin{equation}
\label{rightsand}
g_0(T_1)+ g_0(T_2) \geq \Gamma(f).
\end{equation}

Adding together (\ref{t1}) and (\ref{t2}), we get

\begin{equation}
g_0(T_1) + g_0(T_2) + \ds(f_{T_1=0})\cdot\cs(f_{T_1=0}) + \ds(f_{T_2=0})\cdot\cs(f_{T_2=0}) \leq 2\Gamma(f)
\end{equation}

Applying (\ref{rightsand}) gives 

\begin{equation}
\Gamma(f) + \ds(f_{T_1=0})\cdot\cs(f_{T_1=0}) + \ds(f_{T_2=0})\cdot\cs(f_{T_2=0}) \leq 2\Gamma(f)
\end{equation}

Using Equations (\ref{cs}) and (\ref{ds}), and subtracting $\Gamma(f)$ from both sides
of the inequality yields

\begin{equation}
\ds(f)\cdot\cs(f) \leq \Gamma(f).
\end{equation}
\end{proof}

We also have tight bounds on the 0-goal and 1-goal values
of read-once functions.

\begin{corollary}
\label{cor:rof}
Let $f$ be a read-once function.
Then $\Gamma^0(f) = \ds(f)$ and $\Gamma^1(f) = \cs(f)$.
\end{corollary}

\begin{proof}
By Proposition~\ref{bipartiteOR},
$\Gamma^0(f) \leq \ds(f)$, $\Gamma^1(f) \leq cs(f)$, and
$\Gamma(f) \leq \Gamma^0(f) \cdot \Gamma^1(f)$.
Suppose $\Gamma^0(f) < ds(f)$ or $\Gamma^1(f) < cs(f)$. Then
$\Gamma(f) < ds(f) \cdot cs(f)$, contradicting
the result of Theorem~\ref{roftheorem} that says $\Gamma(f) = ds(f) \cdot cs(f)$,
So
$\Gamma^0(f) = \ds(f)$ and $\Gamma^1(f) = cs(f)$.
\end{proof}

\section{Upper Bounds on Goal Value}

We first prove a general upper bound on the goal value of Boolean functions.

\begin{theorem}
For any Boolean function $f$,
$\Gamma(f) \leq 2^n-1$.
\end{theorem}

\begin{proof}
For $b \in \{0,1,*\}^n$, let $\#_{01}(b)$ be the number of 0 and 1 in $b$.

Let $g:\{0,1,*\}^n \rightarrow \mathbb{Z}_{\geq 0}$ be defined by $g(b)=2^n-1$ iff $b$ contains a certificate of $f$ and $g(b)=\sum\limits_{i=1}^{\#_{01}(b)} 2^{n-i}$ otherwise. This function is clearly monotone. The proof of the submodularity is also straightforward.

\end{proof}

We do not know how close the goal value
of a Boolean function can come to this upper bound
of $2^n -1$.
Deshpande et al.\ showed that the function $f(x_1, \ldots, x_n) = x_1 x_2 \vee x_3 x_4 \vee \ldots \vee x_{n-1} x_n$
has goal value at least $2^{n/2}$~\cite{deshpandeHellersteinKletenik}.
By Theorem~\ref{roftheorem}, its goal value is in fact $n/2*2^{n/2}$.
Theorem~\ref{roftheorem} also implies that another read-once function has
higher goal value.

\begin{proposition}
\label{biggerQ}
Let $n$ be a multiple of 3.
The function $f(x_1, \ldots, x_n) = x_1 x_2 x_3 \vee x_4 x_5 x_6 \vee \ldots \vee x_{n-2}x_{n-1}x_n$
has $\Gamma(f) = 3^{n/3}(n/3)$.
\end{proposition}

We note that Deshpande et al.\ also showed a lower bound of $2^{n/2}$ on the goal
value of Boolean function $f(x_1, \ldots, x_n)$, where $n$ is even,
such that the value of $f$ is 1 iff the binary number represented by the first $n/2$
input bits is strictly less than the binary number represented by the last
$n/2$ bits.  For $n=4$, the function can be represented by the following
DNF formula:  $f(x_1, x_2, x_3, x_4) = \neg x_1 \wedge \neg x_1 \neg x_2 x_4 \vee \neg x_2 x_3 x_4$.
It is easy to see that $(0, 0, *, 1)$ is a minimal 1-certificate
of $f$, and $(1, 1, *, *)$ is a minimal 0-certifcate.
Therefore, $f$ has a minterm $S$ and a maxterm $T$ where $|S \cap T| > 1$.
This violates the well-known property of read-once functions we used in the proof of
Theorem~\ref{roftheorem}, and so $f$ is a not read-once function.  
We do not know how to compute the exact
goal value of this function, for arbitrary $n$.  We leave as an open question whether its
goal value is, in fact, higher than the goal value of the function in Proposition~\ref{biggerQ}.

The fundamental facts from Section~\ref{fun} are upper bounds:
$\Gamma(f) \leq \Gamma^1(f) \cdot \Gamma^0(f)$, $\Gamma^0(f) \leq \ds(f)$, and 
$\Gamma^1(f) \leq \cs(f)$.
These are tight for certain functions.
As shown in Proposition~\ref{andorxor}, for the AND function,
$\Gamma(f) = n$, and it is easy to show that
$\Gamma^0(f)=1=\ds(f)$ while $\Gamma^1(f) = n=\cs(f)$.
Therefore, for the AND (or the OR) function,
we have $\Gamma(f) = \Gamma^1(f) \cdot \Gamma^0(f)$ $= \cs(f)\cdot\ds(f)$.
In contrast, by Lemma~\ref{propnonzero} 
and Proposition~\ref{andorxor}, 
when $f$ is the XOR function, $\Gamma(f) = \Gamma^1(f) = \Gamma^0(f)=n$,
while $\ds(f) = \cs(f) = 2^{n-1}$.

We have the following additional upper bound.

\begin{theorem}
\label{lemk}
Let $f$ be a Boolean function that is not identically 1 or 0.
Then for $k \in \{0,1\}$,
$\Gamma^k(f) \leq \Gamma(f)$.
\end{theorem}

\begin{proof}
Let $g:\{0,1,*\}^n \rightarrow \mathbb{Z}_{\geq 0}$ be a goal function for $f$,
and let $Q$ be its goal value.
Without loss of generality, assume $k=1$.
Let $g^1:\{0,1,*\}^n \rightarrow \mathbb{Z}_{\geq 0}$ 
be such that 
for $b \in \{0,1,*\}^n$, 
$g^1(b)=Q-1$ if $b$ contains a $0$-certificate of $f$,
and $g^1(b) = g(b)$ otherwise.
Clearly $g^1(b) = Q$ iff $b$ contains a 1-certificate of $f$.
Using the monotonicity and submodularity of $g$,
and the fact that any extension of a 0-certificate is also a 0-certificate,
it is straightforward to show that $g^1$ is monotone and submodular.
\end{proof}

\section{Goal Value and Other Measures of Boolean Function Complexity}

Goal value does not appear to be equivalent or
nearly equivalent to other previously studied measures
of Boolean function complexity.
We note first that
many such measures (e.g., sensitivity, polynomial threshold degree,
certificate size) have a maximum value of $n$, while goal value can
be exponential in $n$.  Of course, this does not preclude the possibility
that one such value might be closely related to the logarithm of the goal value.

Paterson defined the notation of a {\em formal} complexity measure
for Boolean functions (see ~\cite{Wegener}).
A quantity $m(f)$ is a formal complexity measure
for Boolean functions $f$
if it satisfies
$m(f \vee g) \leq m(f) + m(g)$ for all pairs of Boolean functions $f$ and $g$
defined on a common set of variables. 
Goal value is {\em not} a formal complexity measure.
For example, consider $f(x_1, x_2, x_3, x_4) = x_1 x_2$
and $g(x_1, x_2, x_3, x_4) = x_3 x_4$.
By our results on read-once functions,
$\Gamma(f) = \Gamma(g) = 2$ while
$\gamma(f \vee g) = \ds(f \vee g) \cdot \cs(f \vee g) = 8$. 

We note that Razborov defined complexity
measure $m(f)$ to be {\em submodular}
if it satisfies
$m(f \wedge g) + m(f \vee g) \leq m(f) + m(g)$~\cite{razborov}. 
He also showed that this measure can never be more than $n$,
so clearly goal value is not a submodular 
complexity measure in this sense.

Above, we have shown relationships between goal value, $\ds(f)$, and $\cs(f)$. 
We now relate 0-goal value 
and 1-goal value of a {\em monotone} Boolean function to decision tree rank and 
polynomial threshold degree.

A decision tree 
can be used to compute a monotone submodular set function
$h:2^V \rightarrow  \{0,1,\ldots, d\}$ as
follows: It 
has internal nodes labeled with variables $x_i \in V$.
Each leaf of the tree is labeled with an element of $\{0,1\ldots, d\}$.
In using the tree to compute the value of $h$ on input $S \subseteq V$,
for each internal node labeled $x_i$,
you branch right if $x_i$
is in $S$, and left if not.

The following lemma will be useful.
 
\begin{lemma} 
\label{novalue}
Let $f$ be a monotone Boolean function. 
Then there is a 1-goal function $g$ for $f$ 
that gives no value to 0's, whose goal value is $\Gamma^1(g)$.
Similarly, there is a 0-goal function $g$ for $f$
that gives no value to 1's, whose goal value is $\Gamma^0(g)$.
\end{lemma}

\begin{proof}
We prove this for the 1-goal function; the proof is symmetric for a 0-goal function.
Let $g$ be a 1-goal utility function for $f$ with 1-goal value $\Gamma^1(f)$.
For $b \in \{0,1,*\}$,
let $b'$ be the partial assignment produced from $b$
by changing all 0's in $b$ to $*$'s.
Define a related utility function $h$ such that $h(b) = g(b')$.
It is straightforward to verify that because $f$ is monotone, $h$ is 
also a 1-goal utility function for $f$, with the
same 1-goal value as $g$, namely $\Gamma^1(f)$.
\end{proof}

We now relate monotone submodular set functions and decision trees.

\begin{lemma}
\label{monotonesubmodrank}
If $h:2^V \rightarrow  \{0,1,\ldots, d\}$ 
is a monotone, submodular set function, then there is a decision tree computing $h$ that has rank
$d$. 
\end{lemma}
\begin{proof} 

It was shown in ~\cite{feldmanRepresentation} 
that for submodular $h$ (not necessarily monotone),
there is a decision tree computing $h$ of rank at most $2d$.
Here we give a better bound on rank when monotonicity holds.

The proof is by induction on $n$ and $d$.
It is clearly true if $n=0$ and $d=0$, and in particular,
it is true for $n+d=0$.
Let $s > 0$, and assume true for $n+d < s$.
We now show true for $n+d = s$.

If $h$ is identically equal to a single value in $\{0,1,\ldots,d\}$, then
the decision tree can consist of a single node labeled with that value.

Otherwise, there exists a subset $S \subseteq V$ such that $h(S) \neq h(\emptyset)$.
By monotonicity, $h(S) > h(\emptyset)$.
Let $S$ be such that $h(S)$ is maximized.  If $h(S) < d$, then a decision tree of rank at most $d$
must exist by induction.  Suppose $h(S) = d$.

By submodularity and monotonicity,
there exits $x_i \in S$ such that $h(\{x_i\}) > h(\emptyset)$, so $h(\{x_i\}) \geq 1$.
We construct a decision tree for $h$ by first putting $x_i$ in the root.  
The left subtree of this tree needs to compute the function
$h':2^{V - \{x_i\}} \rightarrow \{0,1, \ldots, d\}$
such that $h'(S) = h(S)$
for all $S \subseteq V-\{x_i\}$.
This is a function of subsets of $n-1$ variables, and function $h'$ is monotone and submodular, and thus by induction, can be computed
by a tree of rank at most $d$.

The right subtree
needs to compute the function $h'':2^{V - \{x_i\}} \rightarrow \{0,1, \ldots, d\}$
such that $h''(S) = h(S \cup \{x_i\})$ for all $S \subseteq V$. 

We now show that there is a decision tree
computing $h''$ that has rank at most $d-1$.
Consider the function $\hat{h}$ defined on $2^{V - \{x_i\}}$
such that for $S \subseteq V$, $\hat{h}(S) = h(S \cup \{x_i\}) - h(\{x_i\})$ for all $x_i$ in $S$.
Since $h$ is monotone, so is $\hat{h}$.
Further, since $h(\{x_i\}) \geq 1$ and $h(S \cup \{x_i\}) \leq d$,
it follows that $\hat{h}$ can be viewed as mapping to $\{0,1, \ldots, d-1\}$.
Function $\hat{h}$ is defined on $n-1$ variables, and since it is monotone and submodular, 
by induction, it can be computed by a tree of rank at most $d-1$.
If modify the
tree by adding the value $h(\{x_i\})$ to the value in each leaf of the tree, the resulting tree will compute $h''$.

We have thus shown that we can construct a tree computing $h$ such that
the two subtrees of the root have rank $d$ and $d-1$ respectively.
Such a tree has rank $d$.
\end{proof}

Using the above, we relate decision tree rank and goal values.

\begin{theorem}
\label{rankdprop}
Let $f$ be a monotone Boolean function, and let $d = \min\{\Gamma^1(f),\Gamma^0(f)\}$.
Then there is a decision tree of rank at most $d$ computing $f$.
There is a
polynomial-threshold function of degree at most $d$ computing $f$.
\end{theorem}
\begin{proof}
Consider a 1-goal function $g$ for a monotone Boolean function $f$, whose goal value is $d$.
By Lemma~\ref{novalue}, we can assume that the value of $g$ only increases on bits that are set to 1, and
stays the same on bits that are set to 0.
Let $V = \{x_1, \ldots, x_n\}$
Define a set function $h:2^V \rightarrow \{0, \ldots, d\}$,
such that for all $S \subseteq V$, $h(S)$ is equal to the value of $g$ on the partial assignment
setting the variables in $S$ to 1, and all other variables to 0.

By Lemma~\ref{monotonesubmodrank}, there is a decision tree computing $h$ with rank at most $d$. 
This same tree must compute the function $g$
(if at an internal node labeled $x_i$, you branch left if $x_i = 0$, and right if $x_i = 1$). 
For each leaf of the tree, change the value of the leaf to 0 if it is labeled with a value in ${0, \ldots, d-1}$,
and to 1 if it is labeled with the value $d$.
Since $d$ is the 1-goal value for $g$, on any assignment in $2^V$, the value computed by the tree will be 1
iff the assignment contains a 1-certificate for $f$.  It follows that a tree of rank at most $d$ computes $f$.

By~\cite{blumRank}, a function that can be computed by a tree of rank $d$ has an equivalent $d$-decision list. 
By a similar argument to that used in~\cite{ehrenfeuchtGeneral},
any function expressible by a $d$-decison list is also 
expressible as a degree-$d$ polynomial threshold function.

\end{proof}

By the above theorem, 
decision tree rank and threshold degree are lower bounds for
$\min\{\Gamma^1(f),\Gamma^0(f)\}$ for monotone functions.
However, these lower bounds cannot be more than $n$, so they are very
weak in cases where $\Gamma^1(f)$ and $\Gamma^0(f)$ are exponential in $n$.
For example, it is easy to construct a read-once function $f$ 
with exponential values for $\ds(f)$ and $\cs(f)$,
and by Corollary~\ref{cor:rof},
these values are equal to $\Gamma^0(f)$ and $\Gamma^1(f)$.


\section{Integer Program for Computing Goal Value}

The goal value for a Boolean function can be computed
(at least in principle) by solving an integer programming
problem.

Let us first recall the extension graph. Its vertices are
the partial assignments, and it has an edge $p \rightarrow q$
if $q$ can be obtained from $p$ by changing one $\ast$ to a
0 or a 1.  We imagine it to be drawn so that the more ``informative''
(higher weight) vertices are at the top.

Since a partial assignment is just
a length $n$ string over the alphabet $\{0,1,\ast\}$, it has
$3^n$ vertices.  


Let us now count the edges.  Each weight $k$ vertex has
$2(n-k)$ edges connecting to a vertex of weight one higher (pick
the variable to reveal, and then there are two ways to set it).
Therefore, the number of edges is
\begin{equation}
       2 \sum_{k=0}^{n-1} {n \choose k} 2^k (n-k) = 2n \cdot 3^{n-1}
\label{eq:edgecount}
\end{equation}
(To get the right hand side, rewrite the binomial as
${n \choose n-k} = n/(n-k) {n-1 \choose n-k-1} = n/(n-k) {n-1 \choose k}$.)

The ratio of edges to vertices is $\Theta(n)$, similar to the hypercube.

With this picture, we can imagine a goal function $g$ as a labeling
of this graph's vertices by integers, with 0 at the bottom (weight 0) vertex.
(We may as well take the minimum value to be 0.)

Monotonicity means that $g$ increases, or at least does not decrease,
along any upward path.

Submodularity is a little harder to picture.  Imagine a diagram
\begin{eqnarray}
                    &          & \beta'    \nonumber               \\
                    & \nearrow & \uparrow  \nonumber               \\
           \beta    &          & \alpha'   \label{eq:smalldiagram} \\
           \uparrow & \nearrow &           \nonumber               \\
           \alpha   &          &           \nonumber
\end{eqnarray}
where the vertical lines represent paths and the diagonal lines
are edges.  The diagonal lines are ``parallel'' in the sense that the
same variable is set, in the same way.  Then submodularity requires
that
$$
          g(\beta') - g(\beta) \le g(\alpha') - g(\alpha) .
$$
That is, for the alternating sum around the cycle (clockwise or counterclockwise),
\begin{equation}
          g(\beta') - g(\alpha') + g(\alpha) - g(\beta) \le 0.
\label{eq:cycle}
\end{equation}

One can see, by putting together small diagrams such as
\begin{eqnarray*}
                    &          & \beta'    \\
                    & \nearrow & \uparrow  \\
           \beta    &          & \gamma'   \\
           \uparrow & \nearrow & \uparrow  \\
           \gamma   &          & \alpha'   \\
           \uparrow & \nearrow &           \\
           \alpha   &          &        
\end{eqnarray*}
(The contributions at $\gamma$ and $\gamma'$ will cancel out.) that
it is enough, for submodularity, for the cycle property (\ref{eq:cycle}) to
hold for all ``small'' diagrams (\ref{eq:smalldiagram}) in which 
$\beta$ extends $\alpha$ by revealing the value of $x_i$, 
and $\alpha'$ extends $\alpha$
by revealing the value of $x_j$, $j \ne i$.

Using this, we can estimate the number of constraints needed for
an integer program for the minimal goal value.  We may assume that
$f$ is non-constant (otherwise, taking $g$ identically 0 gives the
optimal goal function).

Monotonicity must hold along all edges, which accounts for
$2n 3^{n-1}$ constraints, since there are this many edges.

Submodularity (the cycle property) must hold for all small diagrams.
When $n \ge 2$, the number of small diagrams is
$$
4 \sum_{k=0}^{n-2} {n \choose k} {n-k \choose 2} 2^k
= 2n(n-1) 3^{n-2}.
$$

Finally, there are function value constraints.
The values on certificates must equal $Q$,
and the values on non-certificates must be $\le Q-1$.
In particular, the value on the empty assignment must be
0.  Since there is a constraint for each vertex,
there are $3^n$ of these.

The integer program seeks to minimize $Q$.

We have therefore specified an integer program with
$3^n + 1$ variables and 
$$
2n3^{n-1} + 2n(n-1)3^{n-2} + 3^n = \Theta(n^2 3^n) 
$$
constraints.

The following table gives exact values for the above expressions.

$$
\begin{array}{crr}
                n  & \hbox{\# variables}  & \hbox{\# constraints }  \\
                   &                      &                         \\
                1  &     4                &     5                   \\
                2  &    10                &     25                  \\
                3  &    28                &     117                 \\
                4  &    82                &     513                 \\
                5  &   244                &     2133                \\
                6  &   730                &     8505                \\
                7  &  2188                &     32805               \\
                8  &  6562                &     123201              \\
                9  & 19684                &     452709              \\
               10  & 59050                &     1633689        
\end{array}
$$

(In fact, the integer program could be made somewhat smaller.  For example, we
don't need a separate variable for $Q$.)

We do not know how ``hard'' such integer programs are
to solve.  Since the constraints all have $\pm 1$ coefficients,
it may be of a special type that is tractable.  It could also be
used ``as is'' for very small $n$.  


\section{Conclusion}
In this paper, we have shown a number of results concerning the goal value
of Boolean functions.  Nevertheless, fundamental questions remain open, such as the
uniqueness of goal value functions, the maximum possible goal value of a Boolean function,
and the expected goal value of a random Boolean function.
While we have exact bounds on the goal value of read-once 
and $k$-of-$n$ functions, there are other standard classes
of Boolean functions for which we do not have such bounds.
We expect further work on goal value to yield interesting results connecting goal value to other
complexity measures, and anticipate new connections to other algorithmic problems involving Boolean functions.

\section*{Acknowledgments}
L. Hellerstein was partially supported by NSF Award IIS-1217968.
E. Bach was partially supported by NSF Award CCF-1420750.
We thank Mike Saks for his observation that
goal value can be computed by an integer program.

\section*{References}
\bibliographystyle{elsarticle-harv}
\bibliography{goaljourn}

\begin{thebibliography}{16}
\expandafter\ifx\csname natexlab\endcsname\relax\def\natexlab#1{#1}\fi
\expandafter\ifx\csname url\endcsname\relax
  \def\url#1{\texttt{#1}}\fi
\expandafter\ifx\csname urlprefix\endcsname\relax\def\urlprefix{URL }\fi

\bibitem[{Allen et~al.(2014)Allen, Hellerstein, Kletenik, and
  \"{U}nl\"{u}yurt}]{allenHellersteinUnluyurt}
Allen, S., Hellerstein, L., Kletenik, D., \"{U}nl\"{u}yurt, T., 2014.
  Evaluation of {DNF} formulas. CoRR abs/1310.3673v3.

\bibitem[{Blum(1992)}]{blumRank}
Blum, A., 1992. Rank-r decision trees are a subclass of r-decision lists.
  Information Processing Letters 42~(4), 183--185.

\bibitem[{Buhrman and Wolf(2002)}]{buhrmanWolf}
Buhrman, H., Wolf, R.~D., 2002. Complexity measures and decision tree
  complexity: A survey. Theoretical Computer Science 288~(1), 21--43.

\bibitem[{Crama and Hammer(2011)}]{cramaHammer}
Crama, Y., Hammer, P.~L., 2011. {B}oolean Functions: Theory, Algorithms, and
  Applications. Cambridge University Press.

\bibitem[{Deshpande et~al.(2016)Deshpande, Hellerstein, and
  Kletenik}]{deshpandeHellersteinKletenik}
Deshpande, A., Hellerstein, L., Kletenik, D., 2016. Approximation algorithms
  for stochastic {B}oolean function evaluation and stochastic submodular set
  cover with applications to {B}oolean function evaluation and min-knapsack.
  ACM Transactions on Algorithms 12, 42:1--42:28.

\bibitem[{Ehrenfeucht and Haussler(1989)}]{ehrenfreucht}
Ehrenfeucht, A., Haussler, D., 1989. Learning decision trees from random
  examples. Information and Computation 82~(3), 231 -- 246.

\bibitem[{Ehrenfeucht et~al.(1989)Ehrenfeucht, Haussler, Kearns, and
  Valiant}]{ehrenfeuchtGeneral}
Ehrenfeucht, A., Haussler, D., Kearns, M., Valiant, L., 1989. A general lower
  bound on the number of examples needed for learning. Information and
  Computation 82~(3), 247--261.

\bibitem[{Ehrenfeucht et~al.(1982)Ehrenfeucht, Kahn, Maddux, and
  Mycielski}]{ehrenfeucht}
Ehrenfeucht, A., Kahn, J., Maddux, R., Mycielski, J., 1982. On the dependence
  of functions on their variables. Journal of Combinatorial Theory, Series A
  33~(1), 106--108.

\bibitem[{Feldman et~al.(2013)Feldman, Kothari, and
  Vondr{\'a}k}]{feldmanRepresentation}
Feldman, V., Kothari, P., Vondr{\'a}k, J., 2013. Representation, approximation
  and learning of submodular functions using low-rank decision trees. In:
  Proceedings of the 26th Annual Conference on Learning Theory. pp. 711--740.

\bibitem[{Golovin and Krause(2011)}]{golovinKrause}
Golovin, D., Krause, A., 2011. Adaptive submodularity: Theory and applications
  in active learning and stochastic optimization. Journal of Artificial
  Intelligence Research 42, 427--486.

\bibitem[{Guillory and Bilmes(2011)}]{guilloryBilmesuai11}
Guillory, A., Bilmes, J., 2011. Active semi-supervised learning using
  submodular functions. In: Proceedings of the 27th Conference on Uncertainty
  in Artificial Intelligence. pp. 274--282.

\bibitem[{Harrison(1965)}]{harrison}
Harrison, M.~A., 1965. Introduction to Switching and Automata Theory.
  McGraw-Hill.

\bibitem[{Razborov(1992)}]{razborov}
Razborov, A., 1992. On submodular complexity measures. In: London Math. Soc.
  Lecture Note Series, 169. Cambridge University Press, pp. 76--83.

\bibitem[{Rivest(1987)}]{rivestDecisionList}
Rivest, R.~L., 1987. Learning decision lists. Machine Learning 2~(3), 229--246.

\bibitem[{Salomaa(1963)}]{salomaa}
Salomaa, A., 1963. On essential variables of functions, especially in the
  algebra of logic. Suomalainen Tiedeakatemia.

\bibitem[{Wegener(1987)}]{Wegener}
Wegener, I., 1987. The Complexity of {B}oolean Functions. Teubner.

\end{thebibliography}

\end{document}